\documentclass[nonacm, sigconf]{acmart} 
\setcopyright{none}
\AtBeginDocument{%
  }

\usepackage{amsthm, amsmath, mathtools}
\usepackage{algorithm, algorithmic}
\usepackage{subcaption}
\usepackage{multirow}
\usepackage{color}

\author{Linh Tran}
\email{tranl3@rpi.edu}
\affiliation{%
  \institution{Rensselaer Polytechnic Institute}
  \city{Troy}
  \state{New York}
  \country{USA}
}

\author{Ana Milanova}
\email{milanova@cs.rpi.edu}
\affiliation{%
  \institution{Rensselaer Polytechnic Institute}
  \city{Troy}
  \state{New York}
  \country{USA}
}

\author{Stacy Patterson}
\email{sep@cs.rpi.edu}
\affiliation{%
  \institution{Rensselaer Polytechnic Institute}
  \city{Troy}
  \state{New York}
  \country{USA}
}

\begin{document}

\title{Improving Parameter-Efficient Federated Learning with Differentially Private Refactorization}

\begin{abstract}
  Federated Learning (FL) with parameter-efficient fine-tuning, such as Low-Rank Adaptation (LoRA), enables scalable model training on distributed data. However, when combined with Differential Privacy (DP), LoRA often introduces errors during global aggregation and amplifies the negative effect of DP noise. Existing cross-silo FL approaches mitigate the aggregation error by freezing one LoRA module and applying output perturbation. However, in a restricted low-rank subspaces, this additive noise frequently overwhelms the signals of the weight matrices, leading to suboptimal accuracy. To address this vulnerability, we propose \textbf{FedPower}, a differentially private cross-silo FL framework that reshapes server-side aggregation. Instead of perturbing mismatched low-rank factors, FedPower explicitly reconstructs and clips full-rank client updates to bound the sensitivity. The server then projects the exact aggregated update back into a secure low-rank space using \textbf{PowerDP}, a novel differentially private low-rank factorization mechanism. Based on simultaneous subspace iteration, PowerDP injects calibrated DP noise prior to the final orthonormalization step, effectively mitigates the negative effect of DP noise by preserving matrix orthogonality.
  
  We provide rigorous theoretical analyses establishing sensitivity bounds for subspace projections, proving that FedPower achieves both sample-level and client-level DP. Extensive experiments on various language understanding tasks in cross-silo FL settings show that FedPower is robust against tight privacy budgets while adding negligible computational overheads. Additional empirical study on different DP noise injection schemes validates the effectiveness of PowerDP in improving the tradeoff in accuracy and privacy. Evaluation on three different membership inference attacks validates the robustness and privacy-preserving capability of the proposed framework.
\end{abstract}

\maketitle

\section{Introduction}
\label{sec:intro}

The unprecedented capabilities of Large Language Models (LLMs) have driven their widespread adoption across highly sensitive domains, such as healthcare and finance. To leverage specialized, decentralized datasets (e.g., medical records, private financial transactions) without requiring organizations to directly share raw data, Federated Learning (FL) has emerged as a vital paradigm \cite{pmlr-v54-mcmahan17a}. However, deploying FL for modern LLMs with billions of parameters presents massive computational and communication bottlenecks. To mitigate this, Low-Rank Adaptation (LoRA) \cite{hu2021lora} has become the standard for efficient federated fine-tuning. By freezing the base model and updating only a pair of low-rank matrices, Federated LoRA, significantly reduces client-side computational costs and network communication overhead. We term this general FL framework that uses LoRA as FedLoRA for convenience.

While FL is designed to protect the client data by keeping raw data on local devices and only sharing model updates with a central server, federated LLMs remain vulnerable to privacy breaches \cite{GeipingBD020, MahendranV15}. It is well-established that LLMs can memorize their training data, making them highly susceptible to membership inference attacks \cite{shokri2017membership}. In FL, once the final fine-tuned global model is publicly released, honest-but-curious downstream users can interact with the LLM to execute membership inference attacks, effectively recovering the sensitive local samples originally held by participating clients. Consequently, relying solely on standard FL or FedLoRA is insufficient to protect training samples. It is necessary to adapt cryptographic or privacy frameworks in data-sensitive applications.

To provide mathematically rigorous guarantees against such threats, Differential Privacy (DP) \cite{dworkdp} is widely recognized as the gold standard. In this work, we focus on a Global DP framework, assuming a trusted server whose goal is to privatize and release a secure final model. In a Differentially Private FL (DP-FL) utilizing Global DP, the client model updates is clipped to bound the sensitivity, which can be done either by the clients or the server. Then, the server injects calibrated Gaussian noise directly into the global aggregate before updating the global model, achieving Global DP.
This global perturbation masks the contribution of any individual training example, ensuring that the final published LLM remains private against downstream membership inference attacks. However, standard full-parameter DP-FL requires clients to compute and transmit billion-parameter model updates, which can be expensive in terms of both computation and communication for resource-constrained participants.

\begin{figure*}[t]
    \centering
    \includegraphics[width=0.9\linewidth]{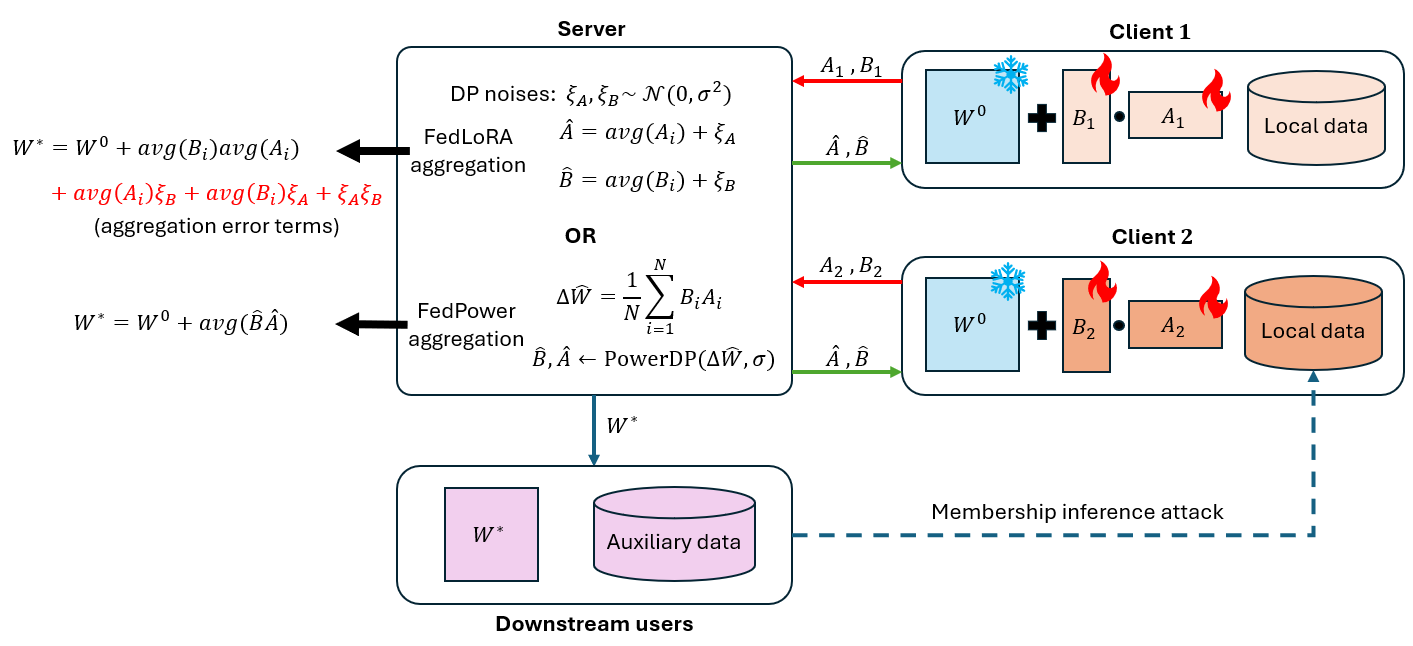}
    \caption{An illustration of the FL algorithms with LoRA in a global DP setting: existing FedLoRA and our proposed method FedPower. In both frameworks, each client trains individual LoRA modules locally and sends to server. In FedLoRA, the server aggregates each LoRA modules separately and adds Gaussian noise, and the resulting global weight contains extra error terms which degrade model accuracy. In FedPower, the server aggregates the products of each LoRA pairs and jointly privatizes and refactorizes the new LoRA modules via PowerDP. At the end of FL training, the server releases the trained global model to downstream users who may perform membership inference attack on the client's private data.}
    \Description{An illustration of the FL algorithms with LoRA in a global DP setting: existing FedLoRA and our proposed method FedPower. In both frameworks, each client trains individual LoRA modules locally and sends to server. In FedLoRA, the server aggregates each LoRA modules separately and adds Gaussian noise, and the resulting global weight contains extra error terms which degrade model accuracy. In FedPower, the server aggregates the products of each LoRA pairs and jointly privatizes and refactorizes the new LoRA modules via PowerDP. At the end of FL training, the server releases the trained global model to downstream users who may perform membership inference attack on the client's private data.}
    \label{fig:fedlora}
\end{figure*}

Naturally, combining Global DP with FedLoRA presents an ideal solution to satisfy both efficiency and privacy. However, simply incorporating DP into standard FedLoRA frameworks results in severe utility degradation due to the mathematical and structural conflicts \cite{sun2024improving}. First, there is an inherent mismatch in the aggregation of LoRA modules. In a precise federated averaging scheme, the server would aggregate the true weight updates from clients, computing $Average(B_i \cdot A_i)$. In contrast, FedLoRA independently averages the low-rank modules to compute $Average(B_i) \cdot Average(A_i)$. Because the average of products does not equal the product of averages, FedLoRA aggregation introduces additional mathematical error terms (cross-terms) into the global update. This mathematical mismatch leads to amplified DP noise during aggregation. When Global DP noise is injected on top of this mismatch aggregation process, the privacy noise interacts with and compounds the error terms from the FedLoRA aggregation as shown in Figure~\ref{fig:fedlora}. This interaction significantly amplifies the negative impact of the DP noise, which leads to severe accuracy degradation in the final fine-tuned LLM.

To bridge this gap and balance the tradeoff between privacy, efficiency, and accuracy, we present \textbf{FedPower}, a robust DP-FL framework tailored for LLMs that entirely avoids noisy low-rank aggregation. Instead of averaging mismatched local LoRA modules, the server aggregates the merged (full-rank equivalent) updates and subsequently performs a private reparameterization, or refactorization, back into the low-rank space. To estimate this low-rank reparameterization while enforcing data privacy, we propose \textbf{PowerDP}, a novel differentially private low-rank refactorization mechanism.

Unlike standard output perturbation techniques that naively add noise to the final matrices, PowerDP seamlessly embeds the DP noise injection process directly into a lightweight power iteration sequence. This mitigates the negative impact of noise on the principal singular vectors, preserving the signals of the LLM updates while satisfying rigorous DP guarantee. Through FedPower, we demonstrate that highly efficient federated LLM fine-tuning with LoRA can be achieved without compromising on the privacy guarantees or degrading the model accuracy.

Our main contributions are:
\begin{enumerate}
    \item We introduce \textbf{FedPower}, a framework that resolves the structural aggregation mismatch in standard Federated LoRA. By shifting the server-side aggregation to the merged, full-rank equivalent updates prior to low-rank refactorization, FedPower eliminates the cross-term mathematical errors. This prevents the compound amplification of DP noise, providing a superior privacy-accuracy tradeoff for decentralized LLM fine-tuning.
    \item We propose \textbf{PowerDP}, a novel differentially private low-rank refactorization algorithm. Rather than naively applying output perturbation to the final LoRA matrices, PowerDP embeds calibrated Gaussian noise directly into a power iteration sequence. This design effectively preserves the principal singular vectors and signals of the model updates while maintaining the same level of privacy guarantee.
    \item We provide a comprehensive theoretical analysis of our framework, proving that PowerDP tightly bounds the $\ell_2$-sensitivity of the subspace projections. We mathematically demonstrate that PowerDP and FedPower satisfies $(\epsilon, \delta)$-Differential Privacy, providing a provably secure defense against downstream membership inference attacks.
    \item We conduct extensive experiments on four language understanding tasks of the GLUE benchmark \cite{wang2018glue}. The empirical results show that FedPower achieves high model accuracy and outperforms existing DP-FedLoRA baselines, while maintaining negligible computational overhead at the server and adding absolutely no extra computational or communication cost to resource-constrained clients.
    \item We evaluate the privacy-preserving performance of FedPower against three different membership inference attacks. Across all levels of privacy budgets, all attacks perform no better than random guessing. This demonstrates that FedPower successfully protect the private training samples from membership inference attacks.
\end{enumerate}

The rest of the paper is organized as follows. In Section~\ref{sec:related}, we discuss the related work on private low-rank factorization and private FL with LoRA. In Section~\ref{sec:background}, we give an overview of Federated Learning and Differential Privacy, and we discuss the challenge of DP-FL with LoRA. In Section~\ref{sec:method}, we describe our proposed methods, the private PowerDP mechanism and FedPower algorithm. In Section~\ref{sec:analysis}, we rigorously analyze the privacy guarantee of PowerDP and FedPower and provide detailed proofs. In Section~\ref{sec:exp}, we conduct extensive experimental study to demonstrate the effectiveness and benefit of FedPower and PowerDP. Finally in Section~\ref{sec:conclude}, we conclude the paper with discussion on future work.

\section{Related Work}
\label{sec:related}

\paragraph{Private Low-Rank Factorization.} A long line of research \cite{blocki2012johnson, hardt2012beating, kapralov2013differentially, hardt2014noisy, balcan2016improved} has explored differentially private low-rank approximations in the context of principal component analysis (PCA). However, the objectives of traditional PCA differ fundamentally from those of FL. More recent methods \cite{yu2021large, zhou2021bypassing, pmlr-v134-kairouz21a, yu2021do} employ low-rank approximations to mitigate the impact of DP noise when training large models such as LLMs. Yu et al. \cite{yu2021large, yu2021do} introduce gradient reparametrization techniques with DP output perturbation. Zhou et al. \cite{zhou2021bypassing} propose a private low-rank gradient projection with a focus on empirical risk minimization. Similarly, Kairouz et al. \cite{pmlr-v134-kairouz21a} introduce a private subspace estimation with the same application in empirical risk minimization. While these approaches align more closely with our training objectives, they are all designed for centralized training environments, rendering them inapplicable to the distributed nature of FL.

\paragraph{FL with LoRA.} Zhang et al. \cite{zhang2024towards} was the first to study a straightforward LoRA integration in FL that lets each client train and update their LoRA matrices $A$ and $B$, and the server averages each LoRA matrix separately. However, this approach leads to suboptimal training accuracy due to the inherent mismatch LoRA aggregation. Other works have proposed to accurately aggregate product of LoRA modules \cite{wang2024flora, sun2024improving, singhal2025fedex}. Wang et al. \cite{wang2024flora} stack LoRA matrices, while Singhal et al. \cite{singhal2025fedex} introduces an additional residual matrix to correct the aggregation error. Nevertheless, they do not consider data privacy, and it is not clear how DP can be applied to their shared components. 

\paragraph{Private FL with LoRA.} The straightforward integration of FedLoRA with DP guarantees has been explored in several recent works \cite{liu2025differentially, 11216445}. Other works such as \cite{kang2024federated} focus on a split LoRA learning scenario, while Xu et al. \cite{xu2024dp} focus on FL with adaptive LoRA ranks. Both of these works are designed for different model setups which are fundamentally different from a standard FL framework that we focus on. Notably, FFA-LoRA \cite{sun2024improving} was the first to identify the compounded error arising from mismatched LoRA aggregation and DP noise injection. To circumvent this, the authors proposed freezing one of the LoRA modules. However, this restricts the model's expressiveness and can hinder the convergence rate. In contrast, FedPower shift the aggregation from the low-rank space to the full-rank space before privately factorizing the aggregated full-rank model updates into new LoRA modules via PowerDP for the next global round. Full-rank aggregation completely avoids the mathematical errors arises in FedLoRA, and PowerDP ensures that the new factorized LoRA modules maintain meaningful signals and expressiveness.

\section{Preliminaries}
\label{sec:background}

\subsection{Federated Learning}
\label{sec:fl}
We consider a standard cross-silo FL architecture consisting of $N$ clients and a central server. Let the global dataset be denoted as $D = \cup_{i=1}^N D_i$, where each client $i$ holds a private local dataset $D_i$. The server maintains a global pre-trained LLM parameterized by $W$. 

In standard FL \cite{pmlr-v54-mcmahan17a}, the global model is optimized through an iterative process of local training and global aggregation. At the beginning of each global round (also called communication round), a subset of clients download the current global model $W$. Each client $i$ then aims to minimize its local empirical loss
\begin{align}
    \mathcal{L}_i = \frac{1}{|D_i|} \sum_{(x,y)\in D_i} \ell(W, x, y),
\end{align}
where $\ell$ denotes the task-specific loss function (e.g., cross-entropy loss for sequence-to-sequence generation). To do so, the client performs several iterations of local gradient descent, updating its local model $W_i$ as
\begin{align}
    W_i \gets W_i - \eta \nabla_{W_i} \mathcal{L}_i.
\end{align}

\begin{algorithm}[t]
\caption{FedLoRA (non-private) \cite{zhang2024towards}}
\label{alg:fedlora}
\begin{algorithmic}[1]
    \STATE {\textbf{Input:}} number of global rounds $T$, number of local rounds $L$, global weight $W^0 \in \mathbf{R}^{m \times n}$, rank $r$, learning rate $\eta$.
    \STATE \textcolor{blue}{// Server initializes global LoRA modules.}
    \STATE $A^0 \sim \mathcal{N}(0, \sigma^2)^{r \times n}$.
    \STATE $B^0 \gets \mathbf{0}_{m \times r}$.
    \FOR {global round $t \gets 1 \ldots T$}
        \STATE Server subsamples client set $\mathcal{C}^t$ with sampling rate $q_c$.
        \FOR {client $i \in \mathcal{C}$ in parallel}
            \STATE \textcolor{blue}{// Client initialize local LoRA modules.}
            \STATE $A^{t,0}_i \gets A^{t-1}$
            \STATE $B^{t,0}_i \gets B^{t-1}$.
            \FOR {local round $l \gets 1 \ldots L$}
                \STATE Client subsamples batch $\mathcal{B}^{t,l}$ with sampling rate $q_s$.
                \STATE \textcolor{blue}{// Client computes loss.}
                \STATE $\mathcal{L}_i^{t,l} \gets \frac{1}{| \mathcal{B}^{t,l} |} \sum_{(x,y)\in \mathcal{B}^{t,l}} \ell(A^{t,l-1}_i, B^{t,l-1}_i, x, y)$.
                \STATE \textcolor{blue}{// Client updates local LoRA modules.}
                \STATE $A^{t,l}_i \gets A^{t,l-1}_i - \eta \nabla_a \mathcal{L}_i^{t,l}$
                \STATE $B^{t,l}_i \gets B_i^{t,l-1} - \eta \nabla_b \mathcal{L}^{t,l}_i$.
            \ENDFOR
            \STATE Clients send $A^{t,L}_i, B^{t,L}_i$ to server.
        \ENDFOR
        \STATE \textcolor{blue}{// Server aggregates and updates global LoRA modules.}
        \STATE $\displaystyle A^t \gets  \frac{1}{|\mathcal{C}|} \sum_{i\in \mathcal{C}} A_i^{t,L}$.
        \STATE $\displaystyle B^t \gets \frac{1}{|\mathcal{C}|} \sum_{i\in \mathcal{C}} B_i^{t,L}$.
    \ENDFOR
    \STATE \textcolor{blue}{// Server merges LoRA for final release.}
    \STATE $W^T \gets W^0 + B^T A^T$.
    \STATE {\textbf{Output:}} $W^T$. 
\end{algorithmic}
\end{algorithm}

At the end of the local training phase, participating clients transmit their updated local models $W_i$ back to the server. The server then updates the global model parameters by averaging these local models
\begin{align}\label{equation:fedavg:aggregate}
    W \gets \frac{1}{N} \sum_{i=1}^N W_i.
\end{align}
This iterative process is repeated until the global model reaches convergence.

Applying full-parameter FL to modern billion-parameter LLMs requires the clients to compute, store, and communicate giant model updates, which is prohibitively expensive. To minimize these computational and communication overheads, the system utilizes Low-Rank Adaptation (LoRA) \cite{hu2021lora}, a parameter-efficient fine-tuning technique that freezes the pre-trained model weights and only fine-tune two trainable low-rank modules.

Specifically, for a pre-trained weight matrix $W \in \mathbb{R}^{m \times n}$, LoRA constrains the weight update $\Delta W$ by representing it as the product of two low-rank matrices $B \in \mathbb{R}^{m \times r}$ and $A \in \mathbb{R}^{r \times n}$, where the rank $r \ll \min(m,n)$. During training, $W$ is frozen, and only $A$ and $B$ are updated. To ensure that the initial state of the adapted model remains identical to the pre-trained model, $A$ is initialized randomly from a Gaussian distribution, and $B$ is initialized to zero, yielding $\Delta W = BA = 0$ at the start of training. The forward pass for an input $x$ is expressed as
\begin{align}
h = W x + \Delta W x = W x + B A x.
\end{align}

In a Federated LoRA (FedLoRA) framework, clients initialize local LoRA modules $A_i$ and $B_i$, train them over their private data, and send only these lightweight low-rank matrices to the server. To aggregate them, the server averages on the low-rank modules independently (i.e., computing $\frac{1}{N}\sum B_i$ and $\frac{1}{N}\sum A_i$). The complete training process for standard FedLoRA, from \cite{zhang2024towards}, is shown in Algorithm~\ref{alg:fedlora}.

\subsection{Differential Privacy}
\label{sec:dp}

Differential Privacy (DP) \cite{dworkdp} is a mathematically rigorous
framework that protects an individual data sample via the notion of \textit{adjacent datasets}. Two datasets $D, D' \in \mathcal{D}$ are adjacent if they differ by exactly one sample (e.g., via the addition, removal, or substitution of a single sample). DP restricts how much the output distribution of a randomized algorithm can change when evaluated on adjacent datasets. The standard $(\epsilon, \delta)$-DP is defined as follows.

\begin{definition} \textbf{(Differential Privacy)}
A randomized mechanism $\mathcal{M}: \mathcal{D} \rightarrow \mathcal{R}$ with domain $\mathcal{D}$ and range $\mathcal{R}$ satisfies $(\epsilon, \delta)$-DP if for any two adjacent datasets $D, D' \in \mathcal{D}$ and for any subset of outputs $\mathcal{S}\subseteq \mathcal{R}$, it holds that
\begin{align}
    \text{Pr}[\mathcal{M}(D)\in \mathcal{S}] \leq e^\epsilon \text{Pr}[\mathcal{M}(D')\in \mathcal{S}] + \delta .
\end{align}
\end{definition}
Here, $\epsilon > 0$ is the privacy budget bounding the multiplicative difference in the output distributions, and $\delta \in [0,1)$ is a small relaxation parameter allowing the tight $\epsilon$ bound to fail with probability at most $\delta$.

To satisfy DP for a deterministic function $f: \mathcal{D} \rightarrow \mathbb{R}^d$, a standard approach is to inject calibrated random noise into the true output. The noise is typically generated from a probability distribution such as Laplace or Gaussian. The required variance of noise, or noise scale $\sigma^2$, relies on the data sensitivity and the desired $\epsilon$ budget. For example, the Gaussian noise mechanism achieves $(\epsilon, \delta)$-DP by perturbing the output of $f$ with random noise drawn from a zero-mean multivariate Gaussian distribution
\begin{align}
    \mathcal{M}(D) = f(D) + \mathcal{N}(0, \sigma^2 I) .
\end{align}
The variance of this Gaussian distribution, $\sigma^2$, is referred to as the \textit{noise scale} (also called noise multiplier). The calibration of the noise scale $\sigma$ depends directly on the desired $\epsilon$ budget and the dataset sensitivity $\Delta_2 f$ defined as
\begin{align}
    \Delta_2 f = \max_{D \sim D'} \| f(D) - f(D') \|_2 .
\end{align}
One of the crucial properties of differential privacy is its robustness to post-processing. This property ensures that an adversary cannot degrade the privacy guarantee by applying additional computations or using auxiliary knowledge.

\begin{proposition} \textbf{(Post-Processing)} \label{theorem:post:processing}
Let $\mathcal{M}: \mathcal{D} \rightarrow \mathcal{R}$ be a randomized mechanism that satisfies $(\epsilon, \delta)$-DP. For any arbitrary randomized mapping $g: \mathcal{R} \rightarrow \mathcal{R}'$ that is independent of the private dataset, the composite mechanism $g \circ \mathcal{M}: \mathcal{D} \rightarrow \mathcal{R}'$ also satisfies $(\epsilon, \delta)$-DP.
\end{proposition}

In many practical applications, such as training machine learning models via iterative algorithms, multiple DP mechanisms are applied sequentially to randomly sampled subsets (batches) of a dataset. To bound the cumulative privacy loss, we rely on the moment accountant \cite{abadi2016deep, mcmahan2017learning} which provides a much tighter bound than the standard sequential composition and advanced composition.

\begin{theorem} \textbf{(Moment Accountant)} \label{theorem:advanced}
Let $\mathcal{M}$ be a sequential mechanism applied over $T$ iterations with sample subsampling rate $q$. There exist constants $c_1,c_2$ such that for any $\epsilon < c_1 q^2 T$ and $\delta > 0$, $\mathcal{M}$ satisfies $(\epsilon, \delta)$-DP if
\begin{align}
    \sigma \geq c_2 \frac{q\sqrt{T\log{(1/\delta)}}}{\epsilon} .
\end{align}
\end{theorem}

This theorem demonstrates that the privacy budget $\epsilon$ degrades sub-linearly (proportional to $\sqrt{k}$). On the other hand, the subsampling rate $q \leq 1$ amplifies the privacy as smaller subsampling rate results in better privacy budget.

While standard DP protects individual samples within a database, also referred to as sample-level DP, it may be unsuitable for cross-device FL, where there are thousands to million of clients and each client represent a private user. In such scenarios, the goal is to protect all data belonging to a specific user using the notion of client-level DP. Let a distributed dataset $D$ consist of data from $N$ clients, $D = \cup_{i=1}^N D_i$, where $D_i$ represents the local dataset of client $i$. Two datasets $D$ and $D'$ are considered \textit{client-adjacent} if they differ by the addition or removal of a single client's entire local dataset.

\begin{definition} \textbf{(Client-Level DP)}
A randomized mechanism $\mathcal{M}: \mathcal{D} \rightarrow \mathcal{R}$ satisfies $(\epsilon, \delta)$-client-level DP if for any two client-adjacent datasets $D, D' \in \mathcal{D}$ and for any subset of outputs $\mathcal{S}\subseteq \mathcal{R}$, it holds that
\begin{align}
    \text{Pr}[\mathcal{M}(D)\in \mathcal{S}] \leq e^\epsilon \text{Pr}[\mathcal{M}(D')\in \mathcal{S}] + \delta .
\end{align}
\end{definition}
To achieve client-level DP, the noise scale $\sigma^2$ must be calibrated according to the sensitivity of the mechanism with respect to an entire client's contribution, rather than a single sample.

Our primary goal is to protect the sample data of each client against \textit{downstream inference attacks} (e.g., membership inference attacks \cite{shokri2017membership}) executed on the final released model. We operate under a Global DP trust model. We assume the central server is honest during the training phase; it faithfully executes the aggregation and applies privacy noise. Furthermore, we assume the communication channels between clients and the server are secured. 

The primary adversary is an external downstream user who obtains access to the final fine-tuned LLM ($W_{final}$). The adversary's objective is to determine whether a specific data sample $(x, y)$ was present in a client's dataset $D_i$. To defend against this adversary, our system aims to provide sample-level $(\epsilon, \delta)$-DP. This ensures that the presence or absence of any single data sample does not significantly alter the final released model weights, which bounds the success rate of membership inference attacks.

\subsection{Differentially Private FedLoRA}
\label{sec:dp_fedlora}

While standard FedLoRA provides a highly communication-efficient protocol, it introduces a structural flaw during the aggregation phase. In an ideal full-parameter FL setup, the server aggregates the true effective weight updates from all clients. For LoRA, the exact effective weight update from client $i$ is the full-rank equivalent matrix $\Delta W_i = B_i A_i$. Thus, the true federated average of the model updates should be $\frac{1}{N} \sum_{i=1}^N (B_i A_i)$.

However, as shown in Algorithm~\ref{alg:fedlora}, standard FedLoRA sidesteps this by performing component-wise averaging on the low-rank modules independently. Because the product of averages does not equal the average of products, this straightforward implementation introduces a mathematical discrepancy compared to standard FL on full-rank weights
\begin{align}\label{equation:lora_eq}
    \underbrace{\Big (\frac{1}{N} \sum_{i=1}^N B_i \Big) \cdot \Big( \frac{1}{N} \sum_{i=1}^N A_i \Big)}_{\text{FedLoRA Aggregation}} \neq \underbrace{\frac{1}{N} \sum_{i=1}^N (B_i A_i)}_{\text{Ideal Aggregation}}.
\end{align}
This discrepancy introduces mathematical errors, specifically the unintended cross-terms between different clients' LoRA modules, into the global updates. In non-private settings, these aggregation errors are introduced over multiple global rounds, which can slow down the convergence. However, this mathematical mismatch becomes a critical bottleneck when extended to a DP-FL setting.

When DP noise is injected into the independent low-rank modules, the underlying aggregation error from Equation~\ref{equation:lora_eq} becomes significantly pronounced. As the global model reconstructs the merged weight update ($B_{global} \cdot A_{global}$), the injected DP noise interacts multiplicatively with both the base parameters and the pre-existing cross-term errors. Consequently, this naive combination of DP and FedLoRA results in a substantial reduction in model accuracy, making it ineffective for practical downstream tasks \cite{sun2024improving}.

\section{Methodology}
\label{sec:method}

Addressing the vulnerabilities of prior approaches, we propose \textbf{FedPower} which reshapes the server-side aggregation pipeline. Instead of perturbing the low-rank modules directly, the server reconstructs the full-rank updates, clips them to bound the sensitivity, aggregates them exactly, and finally projects them back into a secure low-rank space using our novel \textbf{PowerDP} mechanism.

\subsection{The PowerDP Mechanism}
A standard approach to ensuring DP for low-rank factorization is direct output perturbation: factorizing the weight matrices and subsequently adding Gaussian noise to the resulting low-rank matrices \cite{vogels2019powersgd, yu2021large}. However, in restricted low-rank subspaces, this additive noise frequently overwhelms the latent signal. Because the low-rank matrices are sensitive to perturbations, standard output perturbation inevitably leads to degradation in the model's learning capability \cite{li2022large, sun2024improving, tran2025privacy}.

Instead of post-hoc perturbation, PowerDP (Algorithm \ref{sec:method:powerdp}) is designed as an in-processing perturbation mechanism. It is built upon the foundation of simultaneous subspace iteration, also known as power iteration \cite{stewart1981simultaneous}, an efficient and effective method for approximating dominant singular vectors. PowerDP iteratively refines the low-rank approximation but uniquely injects the calibrated DP noise directly into the final projection steps.

PowerDP takes as input a full-rank target matrix $W$ with a known bounded norm $C_W$, alongside the target rank $r$, and privacy parameter $\sigma$. 
\begin{enumerate}
    \item \textbf{Subspace Iteration (Lines 2--8):} The algorithm initializes a random projection matrix $Q$. Over $k$ iterations, it alternates between projecting $W$ onto the current row space ($P = W Q^T$) and column space ($A = P^T W$). The periodic orthonormalization of columns and rows prevents numerical instability and ensures the algorithm converges to the dominant singular subspaces of $W$.
    \item \textbf{In-Processing Noise Injection (Lines 9--10):} Rather than returning the exact singular vectors, PowerDP applies Gaussian noise $\mathcal{N}(0, \sigma^2 C_W^2)$ to the final un-normalized projection matrices $\tilde{B}$ and $\tilde{A}$. The variance of this noise is calibrated to the norm bound $C_W$ of the input matrix, ensuring tight sensitivity and achieving the DP guarantees.
    \item \textbf{Structural Preservation (Line 11):} The final orthonormalization of the row-space matrix $\tilde{A}$ occurs after the noise injection. By applying Gram-Schmidt post-perturbation, the mechanism ensures that the resulting matrices remain orthogonal with norm bounded. This mitigates the effect of DP noise observed in traditional output perturbation methods, as the noise is absorbed into the subspace orientation rather than destroying the magnitude of the singular values.
\end{enumerate}

\begin{algorithm}[t]
    \caption{PowerDP}
    \label{sec:method:powerdp}
    \begin{algorithmic}[1]
    \STATE {\textbf{Input:}} matrix $W \in \mathbb{R}^{m \times n}$ (norm bounded by $C_W$), rank $r \ll \min (m, n)$, number of iterations $k$, noise std $\sigma$.
    \STATE Initialize $Q \in \mathbb{R}^{r \times n}$ from Gaussian distribution.
    \FOR {$k$ iterations}
        \STATE $P \gets W Q^T$.
        \STATE $P \gets$ Orthonormalize columns of $P$.
        \STATE $A \gets P^T W$.
        \STATE $Q \gets$ Orthonormalize rows of $A$.
    \ENDFOR
    \STATE $\tilde B \gets W Q^T + \mathcal{N}(0, \sigma^2 C_W^2)$.
    \STATE $\tilde A \gets A + \mathcal{N}(0, \sigma^2 C_W^2)$.
    \STATE Orthonormalize rows of $\tilde A$.
    \RETURN $\tilde A, \tilde B$.
\end{algorithmic}
\end{algorithm}

\subsection{Federated Learning with PowerDP}

Algorithm~\ref{sec:method:fedpower} details our FL framework, \textbf{FedPower}, which achieves a global DP guarantee by integrating PowerDP into the global aggregation phase. While the client-side local training process follows a standard FL protocol (as described in Algorithm~\ref{alg:fedlora}), the crucial difference lies in how the server accurately aggregates and privately refactorizes the LoRA modules.

The FedPower training lifecycle consists of client-side local updates and a privacy-preserving server-side aggregation step.
\begin{enumerate}
    \item \textbf{Client-Side Local Training (Lines 6--16):} At global round $t$, a sampled subset of clients $\mathcal{C}^t$ downloads the global low-rank modules $A^{t-1}$ and $B^{t-1}$. Each client performs $L$ steps of mini-batch stochastic gradient descent on their local dataset. The gradients $\nabla_a \mathcal{L}_i^{t,l}$ and $\nabla_b \mathcal{L}_i^{t,l}$ are used to independently update the local modules $A_i^{t,l}$ and $B_i^{t,l}$.
    \item \textbf{Server-Side Full-Rank Reconstruction (Line 18):} The server explicitly reconstructs the full-rank weight update $\Delta W^t_i = B_i^{t,L} A_i^{t,L}$ for each client.
    \item \textbf{Clipping and Aggregation (Lines 19--20):} Once in the full-rank space, the server clips the matrices $\Delta W_i^t$ using a clipping threshold $C$. This guarantees that the sensitivity of any single client's update is bounded exactly by $C$. The server then computes the average of these clipped full-rank matrices to form the aggregated global update $\Delta W^t$. 
    \item \textbf{Private Refactorization (Line 21):} Finally, the server must project the aggregated full-rank matrix $\Delta W^t$ back down to the target rank $r$ to be transmitted for the next global round. The server passes $\Delta W^t$, the rank $r$, and the clipping threshold $C$ into the \textbf{PowerDP} subroutine. PowerDP returns the newly privatized LoRA modules $A^t$ and $B^t$ for the next global round.
\end{enumerate}

By decoupling the optimization space (which is low-rank on the clients for communication and memory efficiency) from the privacy-bounding space (which is full-rank on the server for exact sensitivity tracking), FedPower simultaneously preserves the efficiency of LoRA and the mathematical guarantees of Differential Privacy, all without sacrificing model accuracy as we show in Section~\ref{sec:exp}.

\begin{algorithm}[t]
    \caption{FedPower}
    \label{sec:method:fedpower}
    \begin{algorithmic}[1]
    \STATE {\textbf{Input:}} number of global rounds $T$, number of local rounds $L$, global weight $W^0 \in \mathbf{R}^{m \times n}$, rank $r$, learning rate $\eta$, noise scale $\sigma$, clipping threshold $C$, number of power iterations $k$.
    \STATE $A^0 \sim \mathcal{N}(0, \sigma^2)^{r \times n}$.
    \STATE $B^0 \gets \mathbf{0}_{m \times r}$.
    \FOR {global round $t \gets 1 \ldots T$}
        \STATE Server subsamples client set $\mathcal{C}^t$ with sampling rate $q_c$.
        \FOR {client $i \in \mathcal{C}$ in parallel}
            \STATE $A^{t,0}_i \gets A^{t-1}$
            \STATE $B^{t,0}_i \gets B^{t-1}$.
            \FOR {local round $l \gets 1 \ldots L$}
                \STATE Client subsamples batch $\mathcal{B}^{t,l}$ with sampling rate $q_s$.
                \STATE $\mathcal{L}_i^{t,l} \gets \frac{1}{| \mathcal{B}^{t,l} |} \sum_{(x,y)\in \mathcal{B}^{t,l}} \ell(A^{t,l-1}_i, B^{t,l-1}_i, x, y)$.
                \STATE $A^{t,l}_i \gets A^{t,l-1}_i - \eta \nabla_a \mathcal{L}_i^{t,l}$
                \STATE $B^{t,l}_i \gets B_i^{t,l-1} - \eta \nabla_b \mathcal{L}^{t,l}_i$.
            \ENDFOR
            \STATE Clients send $A^{t,L}_i, B^{t,L}_i$ to server.
        \ENDFOR
        \STATE \textcolor{blue}{// Server aggregates and privatizes via PowerDP}
        \STATE Merge $\displaystyle \Delta W^t_i \gets B_i^{t,L} A_i^{t,L}$.
        \STATE Clip $\Delta W_i^t$ by $C$.
        \STATE Aggregate $\displaystyle \Delta W^t \gets \frac{1}{|\mathcal{C}|} \sum_{i\in \mathcal{C}} W_i^t$.
        \STATE Refactorize $\displaystyle A^t, B^t \gets \textbf{PowerDP}(W^t, r, k, \sigma)$.
    \ENDFOR
    \STATE {\textbf{Output:}} $W^T$.
\end{algorithmic}
\end{algorithm}

\section{Privacy Analysis}
\label{sec:analysis}

To formally establish the privacy guarantees of FedPower, we proceed in three steps. First, we rigorously bound the $\ell_2$-sensitivity of the matrix projections within the PowerDP subroutine. Second, we formalize the privacy guarantee of a single PowerDP application. Finally, using moment accountant and privacy amplification by subsampling, we provide the privacy guarantees of the FedPower framework at both the sample and client levels.

\subsection{Sensitivity of Subspace Projections}

To calibrate the Gaussian mechanism within PowerDP, we must ensure that the iterative projections do not inflate the sensitivity of the aggregated weight matrices. Lemma~\ref{lemma} establishes that projecting a bounded matrix onto an orthonormal subspace preserves its Frobenius norm bound. This property is the cornerstone of PowerDP's privacy guarantee, as it allows us to safely inject noise into the un-normalized projection steps.

\begin{lemma}[\textbf{Subspace Projection Sensitivity}]\label{lemma}
Let $W\in \mathbb{R}^{m \times n}$ be a matrix with a Frobenius norm bounded by $C_W$ (i.e., $\|W\|_F \leq C_W$). For any projection matrix $P$ with orthonormal columns ($P \in \mathbb{R}^{m \times r}$) or orthonormal rows ($P \in \mathbb{R}^{r \times n}$), the Frobenius norm of the projection $P^TW$ (or $WP^T$) is bounded by $C_W$.
\end{lemma}

\begin{proof}
We use the trace property of the Frobenius norm: $\| M \|_F^2 = \operatorname{Tr} (M^T M)$. Furthermore, by the cyclic property of the trace operator, $\operatorname{Tr}(AB) = \operatorname{Tr}(BA)$.

\textbf{Case 1: $P$ has orthonormal rows.} 
Since $P \in \mathbb{R}^{r \times n}$ has orthonormal rows, $PP^T = I_r$. We want to bound
\begin{align}
    \| WP^T \|_F^2 &= \operatorname{Tr} \left( (WP^T)^T (WP^T) \right) \\
    &= \operatorname{Tr} \left( P W^T W P^T \right) \\
    &= \operatorname{Tr} \left( W^T W P^T P \right).
\end{align}
Because $P P^T = I_r$, the matrix $P^T P \in \mathbb{R}^{n \times n}$ is an orthogonal projection matrix onto the row space of $P$. Its eigenvalues are exactly $1$ (with multiplicity $r$) and $0$. Therefore, $P^T P \preceq I_n$ under the Loewner partial order. Because $W^TW$ is positive semi-definite, we have
\begin{align}
    \operatorname{Tr} \left( W^T W P^T P \right) &\leq \operatorname{Tr} \left( W^T W I_n \right) \\
    &= \| W \|_F^2 \leq C_W^2.
\end{align}
Taking the square root yields $\|WP^T\|_F \leq C_W$.

\textbf{Case 2: $P$ has orthonormal columns.} 
Since $P \in \mathbb{R}^{m \times r}$ has orthonormal columns, $P^T P = I_r$. We want to bound
\begin{align}
    \| P^TW \|_F^2 &= \operatorname{Tr} \left( (P^TW)^T (P^TW) \right) \\
    &= \operatorname{Tr} \left( W^T P P^T W \right) \\
    &= \operatorname{Tr} \left( P P^T W W^T \right).
\end{align}
Similarly, $P P^T \in \mathbb{R}^{m \times m}$ is an orthogonal projection matrix, so $P P^T \preceq I_m$. Since $WW^T$ is positive semi-definite, it follows that
\begin{align}
    \operatorname{Tr} \left( P P^T W W^T \right) &\leq \operatorname{Tr} \left( I_m W W^T \right) \\
    &= \| W \|_F^2 \leq C_W^2.
\end{align}
Taking the square root yields $\|P^TW\|_F \leq C_W$. This concludes the proof.
\end{proof}

\subsection{Privacy Guarantee of PowerDP}

Building upon the sensitivity bound established in Lemma~\ref{lemma}, we formulate the privacy guarantee of a single application of the PowerDP mechanism.

\begin{theorem}[\textbf{PowerDP Privacy}] \label{theorem:powerdp} 
For any target failure probability $\delta_0 > 0$, the PowerDP mechanism (Algorithm~\ref{sec:method:powerdp}) satisfies $(\epsilon_0, \delta_0)$-DP if the noise standard deviation $\sigma$ satisfies
\begin{align}
    \sigma \geq \frac{\sqrt{2 \log(1.25/\delta_0)}}{\epsilon_0}.
\end{align}
\end{theorem}

\begin{proof}
According to Lemma~\ref{lemma}, the $\ell_2$-sensitivity of the target matrices before noise injection ($\tilde{B}$ and $\tilde{A}$) is bounded by $C_W$. PowerDP injects Gaussian noise $\mathcal{N}(0, \sigma^2 C_W^2)$ into these matrices. By the standard properties of the Gaussian noise mechanism \cite{dworkdp}, adding noise scaled by $\sigma \cdot \Delta f$ (where $\Delta f \leq C_W$ is the $\ell_2$-sensitivity) guarantees $(\epsilon_0, \delta_0)$-DP when
\begin{align}
    \sigma \geq \sqrt{2 \log(1.25/\delta_0)} / \epsilon_0.
\end{align}
Because the final orthonormalization step is applied after the noise injection, it acts as a data-independent post-processing step, which by Proposition~\ref{theorem:post:processing} of DP, maintains the same level of privacy guarantee.
\end{proof}

The theorem above establishes the privacy guarantee of our framework by quantifying the exact privacy leakage of a single application of the PowerDP mechanism. By calibrating the standard deviation $\sigma$ of the injected Gaussian noise to the sensitivity of the query, we can achieve a rigorous $(\epsilon_0, \delta_0)$-DP guarantee per step. This base result is necessary for the privacy analysis of FedPower.

\subsection{Privacy Guarantee of FedPower}

In the context of FL, privacy guarantees can be framed at two distinct granularities: \textit{sample-level} DP (protecting the inclusion of a single data sample of a client) and \textit{client-level} DP (protecting the inclusion of all data samples belonging to a specific client). We show that FedPower can accommodate both paradigms. 

\begin{theorem}[\textbf{Sample-Level DP of FedPower}] \label{theorem:fedpower_sample} 
Let $T$ be the total number of global rounds. There exists an absolute constant $c$ such that for any $\epsilon < c q_c^2q_s^2 T$ and $\delta > 0$, Algorithm~\ref{sec:method:fedpower} satisfies $(\epsilon, \delta)$-DP under sample-level privacy if
\begin{align}
    \sigma \geq c \frac{q_c q_s\sqrt{T \log(1/\delta)}}{\epsilon},
\end{align}
where $q_c$ is the probability of a client being sampled in a global round, and $q_s$ is the probability of a sample being drawn in a local batch.
\end{theorem}

\begin{proof}
In FedPower, the global clipping operation ensures that the sensitivity of any single client's full-rank update $\Delta W_i^t$ is bounded by $C$. PowerDP subsequently adds Gaussian noise calibrated to this sensitivity limit $C$. 
For a single training step without subsampling, achieving $(\epsilon, \delta)$-DP requires
\begin{align}
    \sigma \geq \sqrt{2\log(1.25 / \delta)} / \epsilon .
\end{align}
Over $T$ sequential applications, Theorem~\ref{theorem:advanced} states that the privacy loss accumulates at a rate of $\mathcal{O}(\sqrt{T})$. Therefore, the corresponding noise multiplier needs to be scaled by
\begin{align}
    \sigma \geq \sqrt{2T\log(1.25 / \delta)} / \epsilon .
\end{align}

FedPower features a double-sampling mechanism: clients are sampled with rate $q_c$, and data points within clients are sampled with rate $q_s$. The inclusion probability of any single data sample in a given round is $q_c q_s$. Applying the privacy amplification by subsampling, the privacy cost per step is heavily discounted. Therefore, integrating Theorem~\ref{theorem:advanced} with subsampling rate $q_c q_s$ over $T$ iterations yields the final required noise scale
\begin{align}
    \displaystyle \sigma = \Omega\left(\frac{q_c q_s \sqrt{T \log(1/\delta)}}{\epsilon}\right) .
\end{align}
\end{proof}

Building upon the single-step guarantee, Theorem~\ref{theorem:fedpower_sample} provides the overall privacy bound of FedPower under a sample-level threat model.

Sample-level DP is the gold standard in cross-silo FL, which is the setting that we focus on. Other cross-device FL settings with massive number of clients often requires protecting the participation of entire devices. Supporting these scenarios, we further provide the theoretical results on client-level DP of FedPower below.

\begin{theorem}[\textbf{Client-Level DP of FedPower}]\label{theorem:fedpower_client}
Let $T$ be the number of global rounds. There exists an absolute constant $c$ such that for any $\epsilon < c q_c^2 T$ and $\delta > 0$, Algorithm~\ref{sec:method:fedpower} satisfies $(\epsilon, \delta)$-global DP under client-level privacy if
\begin{align}
    \sigma \geq c \frac{q_c \sqrt{T \log(1/\delta)}}{\epsilon}.
\end{align}
\end{theorem}

\begin{table*}[t]
\caption{Test accuracy on $4$ GLUE datasets under different sample-level privacy budgets. We report average accuracy over 5 runs.} \label{sec:exp:table_acc}
\begin{center}
\begin{tabular}{ | c | c | c c c c c c | }
\hline
Privacy budget & Method & \multicolumn{2}{c}{MNLI} & SST-2 & QQP & QNLI & Average \\
& & Match & Mismatch & & & & \\
\hline\hline
& FedLoRA & 94.95 & 89.60 & 84.21 & 85.48 & 85.69 & 87.99 \\
Non-private & FFA-LoRA & 95.07 & 90.52 & 85.25 & 86.35 & 86.58 & 88.75 \\
& FedPower & \textbf{95.33} & \textbf{92.12} & \textbf{86.19} & \textbf{86.89} & \textbf{86.78} & \textbf{89.46} \\
\hline
& FedLoRA & 93.46 & 87.86 & 83.96 & 82.95 & 83.24 & 86.30 \\
$\epsilon = 9$ & FFA-LoRA & 93.69 & 87.61 & 81.26 & 82.19 & 82.90 & 85.53 \\
& FedPower & \textbf{94.04} & \textbf{89.72} & \textbf{84.16} & \textbf{83.61} & \textbf{84.12} & \textbf{87.13} \\
\hline
& FedLoRA & 93.12 & 86.97 & 83.70 & 82.60 & 82.88 & 85.85 \\
$\epsilon = 6$ & FFA-LoRA & 93.23 & 87.52 & 81.47 & 82.21 & 82.75 & 85.44 \\
& FedPower & \textbf{93.72} & \textbf{88.20} & \textbf{84.01} & \textbf{83.12} & \textbf{83.72} & \textbf{86.55} \\
\hline
& FedLoRA & 91.28 & 83.89 & 79.56 & 81.26 & 81.70 & 83.53 \\
$\epsilon = 3$ & FFA-LoRA & 92.43 & \textbf{87.06} & 80.85 & \textbf{81.94} & 82.60 & 84.98 \\
& FedPower & \textbf{93.03} & 84.99 & \textbf{83.30} & \textbf{81.94} & \textbf{82.61} & \textbf{85.18} \\
\hline
\end{tabular}
\end{center}
\end{table*}

\begin{proof}
The proof of Theorem~\ref{theorem:fedpower_client} parallels the proof of Theorem~\ref{theorem:fedpower_sample}. The main distinction is the unit of privacy. In client-level DP, the adjacent datasets differ by an entire client's local dataset. Because Algorithm~\ref{sec:method:fedpower} clips the \textit{entire} updated weight difference $\Delta W_i^t$ from client $i$ to a maximum norm of $C$, the sensitivity of changing one client is exactly $C$. The privacy amplification in this setting stems solely from the client subsampling rate $q_c$ (as local mini-batch sampling $q_s$ does not hide a client's overall participation). Replacing the combined sampling rate $q_c q_s$ with $q$ in Theorem~\ref{theorem:advanced} yields the result.
\end{proof}

Theorem \ref{theorem:fedpower_client} elevates the privacy guarantee to the client level, accommodating the requirements of cross-device federated learning where a user's entire local dataset must be protected simultaneously. Unlike Theorem \ref{theorem:fedpower_sample}, this bound relies exclusively on client-level subsampling ($q_c$) for privacy amplification, reflecting the stronger threat model where an adversary attempts to infer the presence of an entire device rather than a single sample.

\section{Experimental Evaluation}
\label{sec:exp}

In this section, we conduct a comprehensive empirical study to evaluate the performance of FedPower algorithm and PowerDP mechanism. In Subsection~\ref{sec:exp:setup}, we detail the experimental setup. In Subsection~\ref{sec:exp:accuracy}, we evaluate the effectiveness of FedPower compared with the other two existing baselines
on DP-FL with LoRA. In Subsection~\ref{sec:exp:subsec:powerdp}, we demonstrate the effectiveness of PowerDP compared with naive input and output perturbation. In Subsection~\ref{sec:exp:subsec:frequency}, we study the effect of refactorization frequency on the model accuracy and discuss corresponding tradeoff in accuracy and computational cost. In Subsection~\ref{sec:exp:subsec:mia}, we empirically validate the data privacy guarantee of FedPower via membership inference attacks.

\subsection{Experimental Setup}
\label{sec:exp:setup}

We conduct experiments to evaluate the performance of FedPower with language understanding tasks on the RoBERTa large model \cite{liu2019roberta}. We use four datasets from the GLUE benchmark \cite{wang2018glue} and evenly split each dataset across $6$ clients.
\begin{itemize}
    \item Multi-Genre Natural Language Inference (MNLI) \cite{williams2018broad}: a sentence-sentence pair dataset for classification task with three labels. The task determines if two given sentences are supporting, neutral, or contradicting. The dataset can be evaluated on matched (in-domain) and mismatched (out-domain) test subsets.
    \item Quora Question Pairs (QQP) \cite{sharma2019natural}: a question-question pair dataset for paraphrasing classification. The task detects if two given questions are duplicate or not.
    \item Question Natural Language Inference (QNLI) \cite{wang2018glue}: a question-sentence pair dataset for entailment classification task. The task determines whether a sentence answers a given question. 
    \item Stanford Sentiment Treebank v2 (SST-2) \cite{socher2013recursive}: a sentence dataset for sentiment classification. The task determines if a given sentence is positive or negative.
\end{itemize}

We apply LoRA with rank $r=8$ to the query and value projections with dropout rate of $0.05$. All other parameters, including non-LoRA weights and classification heads are frozen during training. In order to make a fair comparison, we keep the same batch size $b=128$ and total global rounds $T=200$ with client sample rate set to $0.5$ for all experiments. We use the Adam optimizer with learning rate $\eta=0.05$ for non-private settings and $\eta=0.5$ for private settings.

In all experiments, we focus on sample-level DP for a cross-silo FL setting consisting of $6$ clients. We consider three privacy budgets $\epsilon=\{9,6,3\}$ with $\delta=10^{-5}$ and clipping threshold $C=2$.

\begin{figure*}[t]
\centering
\begin{subfigure}{0.244\textwidth}
\centering
  \includegraphics[width=\textwidth]{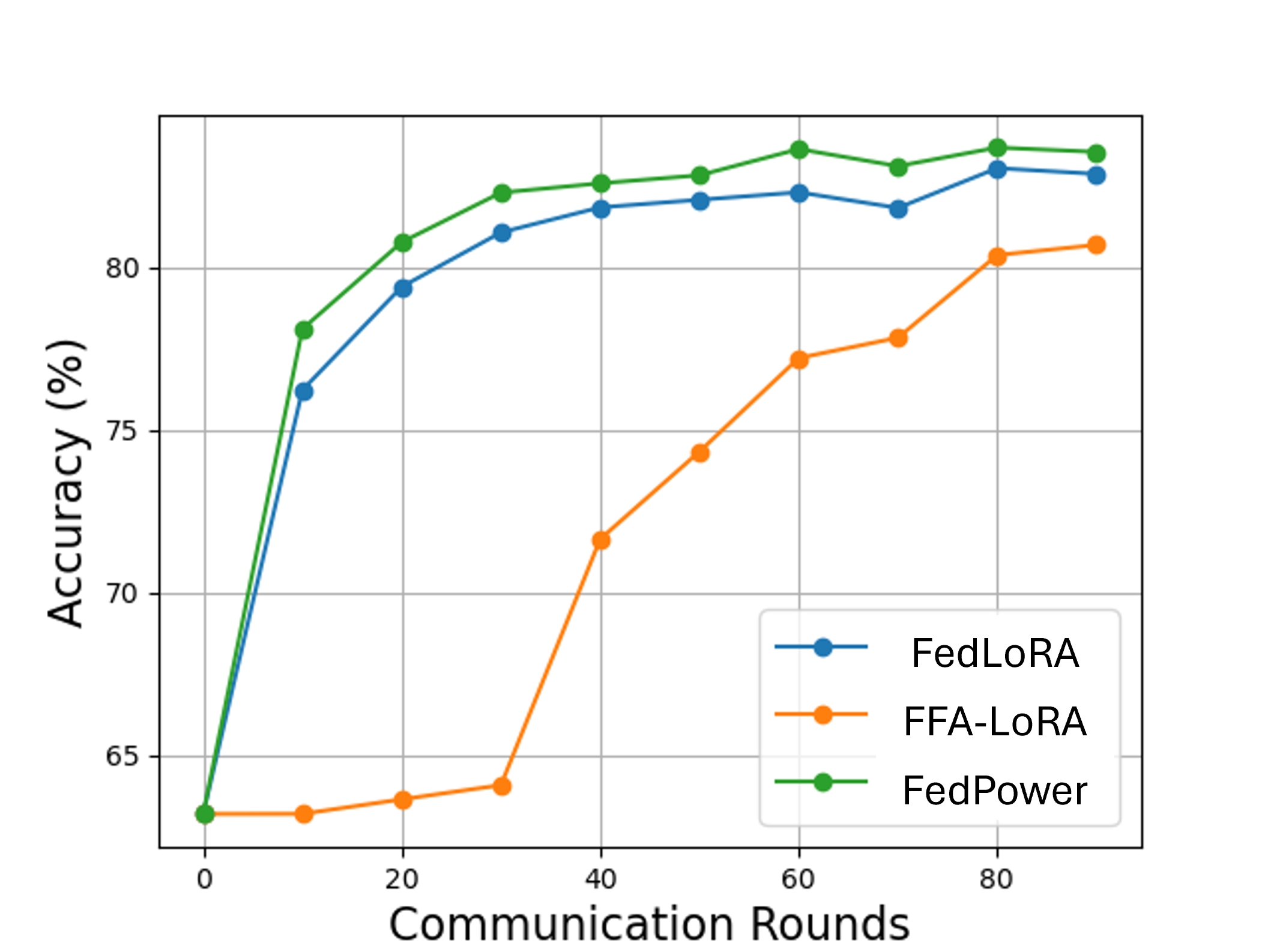}
  \caption{Accuracy vs. rounds ($\epsilon = 6$).}
  \label{sec:exp:round:a}
\end{subfigure}
\begin{subfigure}{0.244\textwidth}
\centering
  \includegraphics[width=\textwidth]{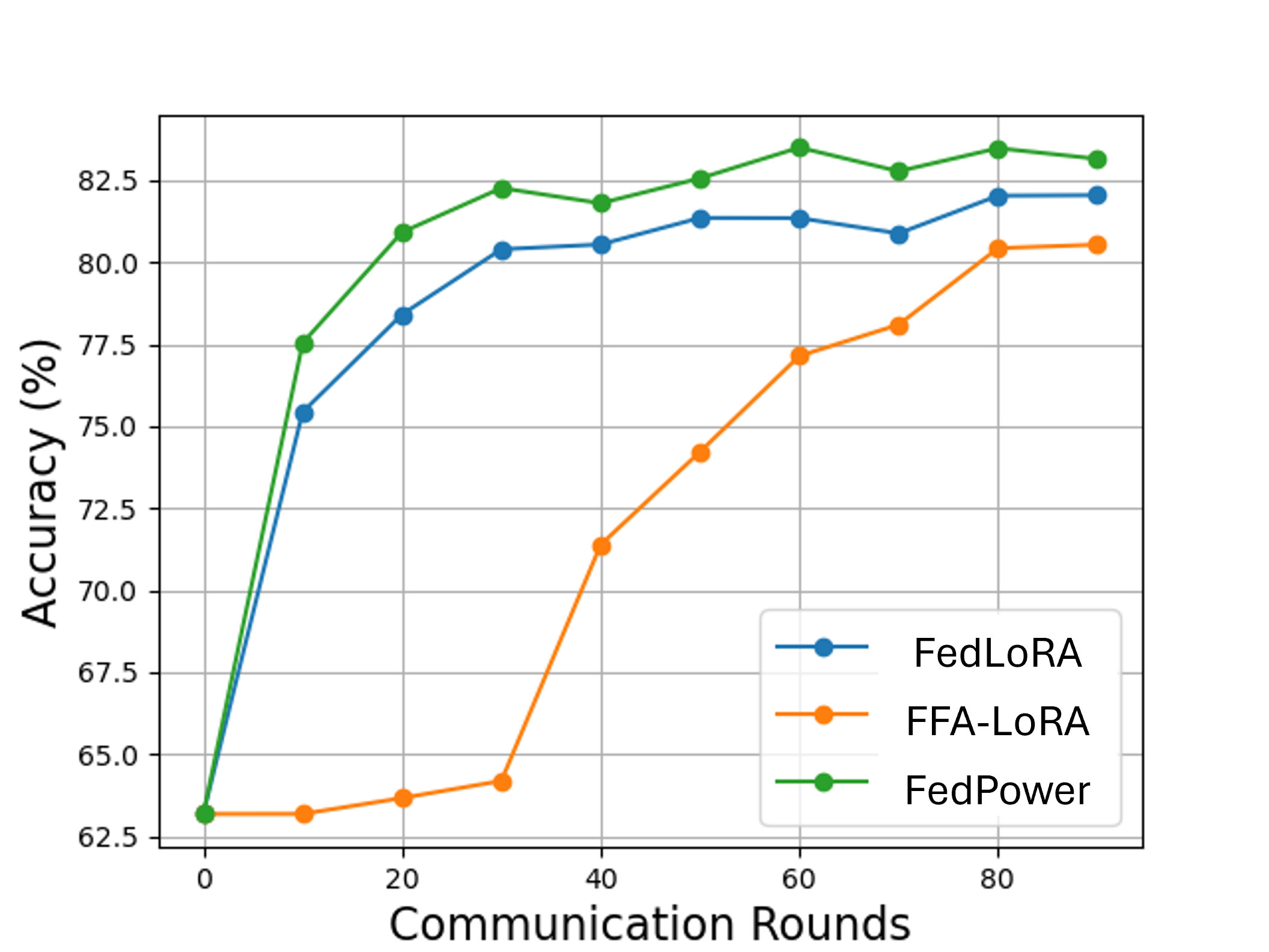}
  \caption{Accuracy vs. rounds ($\epsilon = 3$).}
  \label{sec:exp:round:b}
\end{subfigure}
\begin{subfigure}{0.244\textwidth}
\centering
  \includegraphics[width=\textwidth]{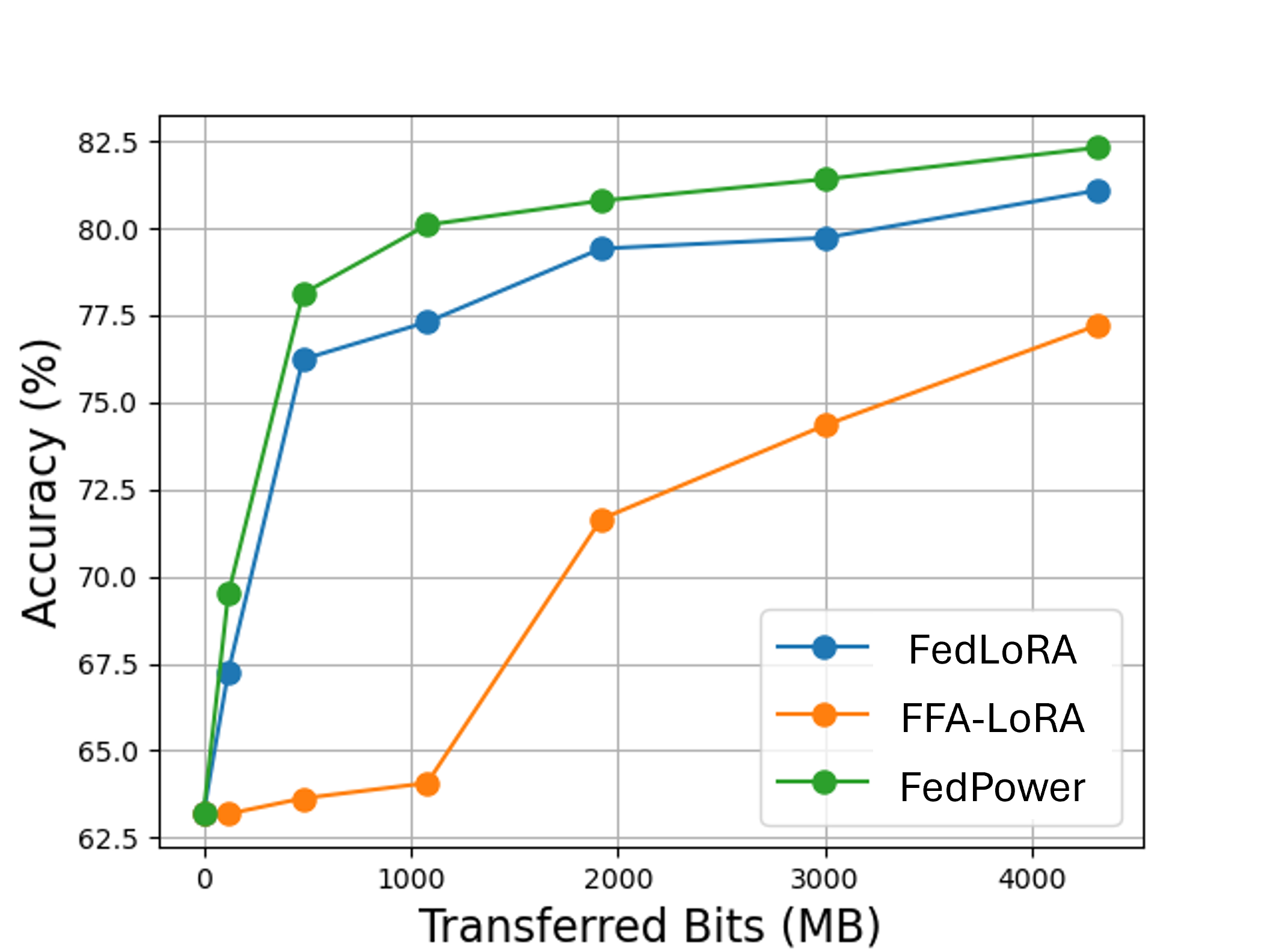}
  \caption{Accuracy vs. bits ($\epsilon = 6$).}
  \label{sec:exp:bit:a}
\end{subfigure}
\begin{subfigure}{0.244\textwidth}
\centering
  \includegraphics[width=\textwidth]{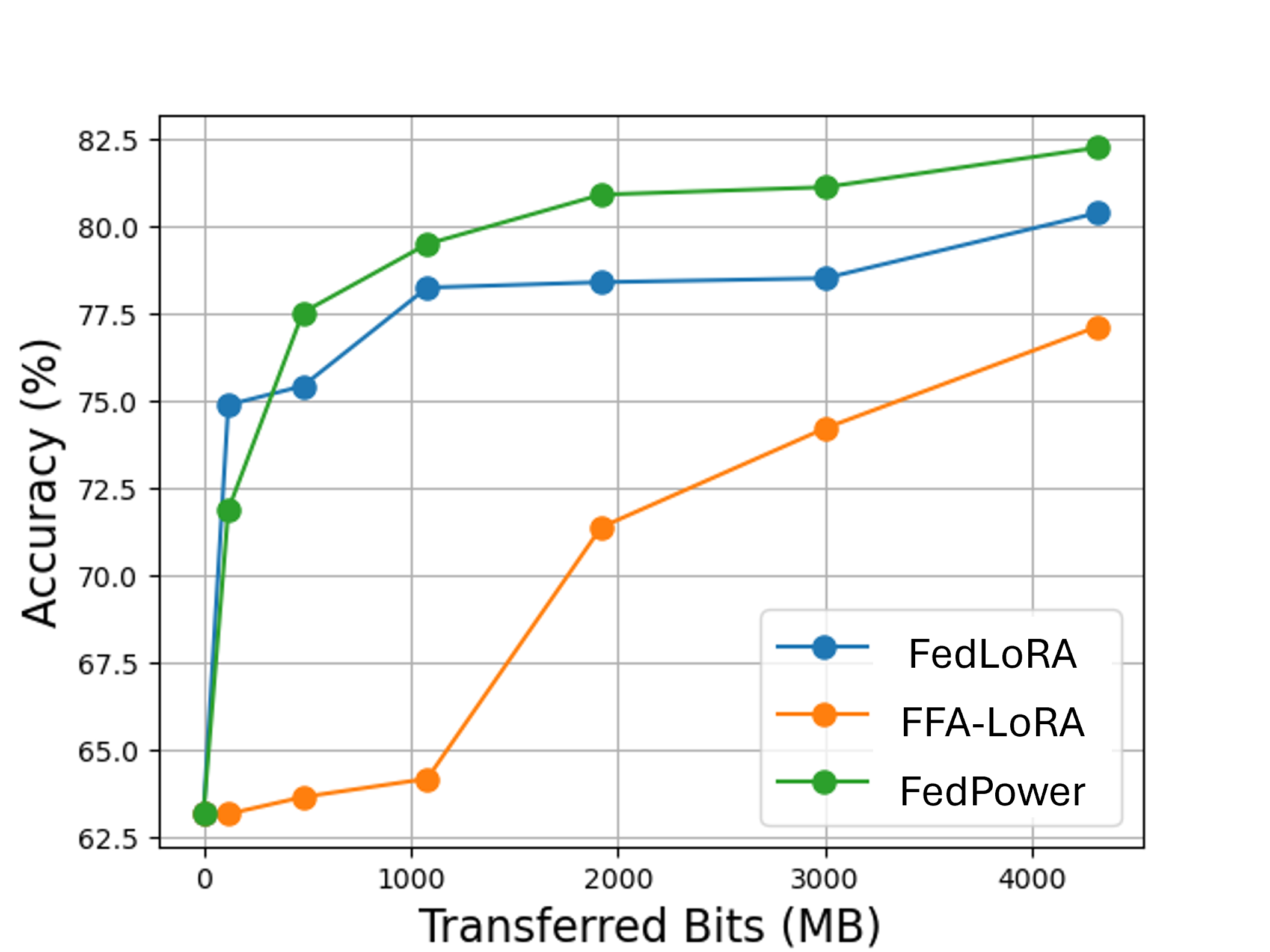}
  \caption{Accuracy vs. bits ($\epsilon = 3$).}
  \label{sec:exp:bit:b}
\end{subfigure}
\Description{Accuracy as a function of global rounds (a and b) and transferred bits (c and d) on QQP with $\epsilon = \{6,3\}$.}
\caption{Accuracy as a function of global rounds (a and b) and transferred bits (c and d) on QQP with $\epsilon = \{6,3\}$.}
\label{sec:exp:convergence_acc}
\end{figure*}

We compare the FL experiment using PowerDP with the following baselines:
\begin{itemize}
    \item \textbf{FedLoRA}: straightforward fine-tuning with LoRA modules, where both of the low-rank components are updated. This method is equivalent to previous works in \cite{zhang2024towards, liu2025differentially, 11216445}. In this case, we directly add DP noise to the aggregated low-rank modules.
    \item \textbf{FFA-LoRA} \cite{sun2024improving}: fine-tuning with LoRA modules, but freezing one low-rank module. Clients only update and share one low-rank module. In this case, we directly add DP noise to only the shared LoRA module.
\end{itemize}

We provide more details of the experiment and the algorithms of the baseline methods in the Appendix. All experiments are done on a cluster of three V100 32GB GPU machines. We repeat all experiments five times using different seeds and report the mean test accuracy across all seeds.

\subsection{Comparison with Existing Baselines}
\label{sec:exp:accuracy}

Table \ref{sec:exp:table_acc} illustrates the privacy-accuracy tradeoff across the GLUE benchmark. In the non-private setting, FFA-LoRA outperforms FedLoRA by $+0.76$ points on average, which can be attributed to reduced aggregation error via freezing one LoRA module. Our proposed FedPower further improves performance by $+0.71$ points (a $+1.47$ point gain over FedLoRA). This highlights the advantage of training and reparametrizing both LoRA modules, overcoming the capacity limitations of FFA-LoRA's restricted learning space.

In private settings, FedPower consistently outperforms all baselines. Because global DP (server-side noise injection) inherently preserves higher baseline accuracy than local DP due to fewer noise addition, achieving consistent gains is particularly challenging. Nevertheless, FedPower establishes a clear advantage, showing the strongest performance under moderate privacy budgets ($\epsilon=9$ and $\epsilon=6$). In addition, FedPower demonstrates strong robustness to noise even under a tight privacy budget ($\epsilon=3$), yielding an average accuracy improvement of $+1.65$ points over FedLoRA and $+0.20$ points over FFA-LoRA.

While FedPower consistently outperforms both baselines across all datasets, FFA-LoRA shows a surprisingly high accuracy ($87.06$) in MNLI out-domain test set. This improvement stems from existing observation that lower learning dimension (from only one LoRA module) improves out-of-domain generalization under the effect of DP noise \cite{tran2025privacy}.

Figure~\ref{sec:exp:convergence_acc} illustrates accuracy as a function of communication rounds and transferred bits for FedLoRA, FFA-LoRA and FedPower. We observe from Figures~\ref{sec:exp:round:a} and \ref{sec:exp:round:b} that FFA-LoRA generally has slower convergence, reaching $80\%$ after $80$ rounds. Meanwhile, FedLoRA requires $30$ rounds and FedPower requires only $20$ rounds to reach the same test accuracy. Even though FFA-LoRA has benefit over communication cost as each client only shares one LoRA module instead of two, Figures~\ref{sec:exp:bit:a} and \ref{sec:exp:bit:b} show that FedLoRA still underperforms both FedLoRA and FedPower given the same limit of bits transferred. In contrast to the suboptimal performance observed from FFA-LoRA, FedPower not only outperforms both baselines with the highest accuracy, but also shows the best tradeoff between accuracy and communication cost as it requires the least the number of communication rounds and transferred bits to reach a desirable target accuracy.

\subsection{The Benefit of In-Processing DP Noise}
\label{sec:exp:subsec:powerdp}
To isolate the specific accuracy-privacy tradeoff benefit of PowerDP, we compare it against alternative DP injection schemes that inject DP noise separately from the factorization process.
\begin{itemize}
    \item \textbf{Input Perturbation:} DP noise is added to the input matrix $W$ before the factorization. By the post-processing properties of differential privacy, the resulting modules $A$ and $B$ inherently satisfy the required privacy guarantees.
    \item \textbf{Output perturbation:} DP noise is added to the output $A$ and $B$ after the factorization, directly protect the outputs with DP. Because the DP noise is applied post-orthonormalization, we calculate the sensitivity of the output and calibrate the noise scale accordingly.
\end{itemize}

For both input and output perturbation, we use the standard non-private power iteration \cite{stewart1981simultaneous} for factorization. We describe the detailed algorithm of input and output perturbation in the Appendix.

As shown in Table \ref{sec:exp:noise_position}, output perturbation results in catastrophic accuracy failure (e.g., $60.59\%$ at $\epsilon=3$). This failure stems from the orthonormalization step in power iteration, which forces the output sensitivity to scale proportionally with the matrix rank, thereby necessitating massive noise injection to satisfy DP guarantees. PowerDP circumvents this by bounding the sensitivity through the projection (Lemma \ref{lemma}), yielding much better privacy-accuracy tradeoff.

\begin{table*}[t]
\caption{Test accuracy of FedPower on $4$ GLUE datasets with different privacy noise injection schemes.}\label{sec:exp:noise_position}
\begin{center}
\begin{tabular}{ | c | c | c c c c c c | }
\hline
Privacy budget & Method & \multicolumn{2}{c}{MNLI} & SST-2 & QQP & QNLI & Average \\
& & Match & Mismatch & & & & \\
\hline\hline
& Input perturbation & 92.61 & 87.88 & 83.27 & 81.26 & 82.81 & 85.56 \\
$\epsilon = 9$ & Output perturbation & 91.40 & 87.41 & 83.26 & 79.83 & 80.45 & 84.47 \\
& PowerDP & \textbf{94.04} & \textbf{89.72} & \textbf{84.16} & \textbf{83.61} & \textbf{84.12} & \textbf{87.13} \\
\hline
& Input perturbation & 92.38 & 87.79 & 82.64 & 81.95 & 81.46 & 85.24 \\
$\epsilon = 6$ & Output perturbation & 89.79 & 83.87 & 82.68 & 52.63 & 53.78 & 72.55 \\
& PowerDP & \textbf{93.72} & \textbf{88.20} & \textbf{84.01} & \textbf{83.12} & \textbf{83.72} & \textbf{86.55} \\
\hline
& Input perturbation & 92.04 & 84.23 & 81.32 & 80.01 & 80.80 & 83.68 \\
$\epsilon = 3$ & Output perturbation & 50.92 & 79.00 & 80.64 & 44.97 & 47.45 & 60.59 \\
& PowerDP & \textbf{93.03} & \textbf{84.99} & \textbf{83.30} & \textbf{81.94} & \textbf{82.61} & \textbf{85.18} \\
\hline
\end{tabular}
\end{center}
\end{table*}

\begin{table*}[t]
\caption{Computational cost and test accuracy of FedPower with various refactorization frequency under $\epsilon=3$.}\label{sec:exp:frequency}
\centering
\small
\begin{tabular}{| c | c c c c |}
\hline
Method & Factorization  & Average aggregation & Computational & Average accuracy \\
& frequency & time per round & overhead & on GLUE \\
\hline\hline
FedLoRA & None & 4.36 s & 0.0 \% & 83.53 \\
FedPower & 10 & 4.41 s & 1.2 \% & 84.37 \\
FedPower & 5 & 4.45 s & 2.2 \% & 85.07 \\
FedPower & 1 & 4.76 s & 9.3 \% & 85.18 \\
\hline
\end{tabular}
\end{table*}

While input perturbation achieves higher accuracy than output perturbation, it still underperforms FedPower equipped with PowerDP. We hypothesize that injecting noise prior to factorization distorts the principal components of the full-rank adaptation matrix, leading to suboptimal singular value estimates. Ultimately, by strategically integrating DP noise during the factorization process, PowerDP yields highly accurate low-rank approximations at an equivalent privacy budget, delivering the best overall tradeoff.

\subsection{Frequency of Refactorization}
\label{sec:exp:subsec:frequency}

We explore the effect of reducing the frequency of refactorization in FedPower. For global rounds that do not perform factorization, we adopt the straightforward aggregation process as described in Equation~\ref{alg:fedlora} (lines 22--23). Similar to FedLoRA, we add DP noise to the resulting aggregated $A$ and $B$.

Table~\ref{sec:exp:frequency} shows the average test accuracy on the GLUE benchmark for which PowerDP is applied every 1, 5, or 10 global rounds. We observe that FedPower with any factorization frequency outperforms FedLoRA on average, demonstrating the increased accuracy benefit of the refactorization process while incurring less than 10\% computational overhead (measured in total runtime). Reducing the frequency of refactorization slightly decreases the test accuracy but comes with less computational overhead. These observations demonstrate the benefit of PowerDP refactorization and its robustness to the choice of factorization frequency. While higher frequency ensures higher accuracy, variants with less frequent refactorization such as every $5$ rounds can offer a favorable tradeoff when computational efficiency is prioritized.

\subsection{Membership Inference Attacks}
\label{sec:exp:subsec:mia}

In this subsection, we aim to demonstrate the effectiveness of FedPower in protecting data samples against various types of membership inference attack. Let $W_{tar}$ be the target released model, and we assume that the attacker has access to an auxiliary dataset $D_{aux}$ which comes from the same data distribution with the private training dataset. We consider three types of inference attacks as described below.

\paragraph{Shadow Model Attack \cite{shokri2017membership}.} The attacker trains $s$ shadow models $W_{sha}^1 \dots W_{sha}^s$ on $s$ non-overlapping subsets of $D_{aux}$. For each shadow model, the attacker records the prediction probability vectors, or confidence scores, for both members (samples used to train $W_{sha}^i$) and non-members. These confidence scores, along with their member or non-member labels, are used to train a set of attack models (binary classifiers). To infer the membership of a target sample $(x,y)$, the adversary feeds $(x,y)$ into $W_{tar}$, extracts the resulting prediction vector, and passes it to the corresponding attack model, which outputs the probability that the sample was a member of the training set.

\paragraph{Loss Based Attack \cite{yeom2018privacy}.} This attack relies on the magnitude of the loss function to predict membership. First, the adversary trains a shadow model $W_{sha}$ on $D_{aux}$ to estimate the training loss of the target model. Then, for each sample $(x, y) \in D_{aux}$, the sample is considered a member of the training set if:
\begin{align}
    \ell (W_{tar}, x,y) < \frac{1}{|D_{aux}|} \sum_{(x', y')\in D_{aux}} \ell (W_{sha}, x,y).
\end{align}

In this case, the averaged prediction loss is used as a threshold to classify if a given sample is part of the training set or not.

\paragraph{Calibration Attack \cite{watson2021importance}.} This attack is an improved version of the loss based attack that addresses the difficulty variance across different data samples. Instead of a global threshold, it uses a per-sample calibration to determine membership. Specifically, the attacker trains $s$ shadow models $W_{sha}^1 \dots W_{sha}^s$ on $s$ non-overlapping subsets of $D_{aux}$. For a target sample $(x,y)$, the adversary observes its loss across the shadow models that did not include $(x, y)$ in their training sets to construct a null distribution of non-member losses. The sample is predicted as a member if its loss on the target model is significantly lower than its expected loss when excluded from training:
\begin{align}
\text{Score}(x,y) = \frac{\ell(W_{tar}, x, y) - \mu_{out}(x,y)}{\sigma_{out}(x,y)} < \tau,
\end{align}
where $\mu_{out}$ and $\sigma_{out}$ are the mean and standard deviation of losses from the out-of-sample shadow models, and $\tau$ is a determined threshold.

\begin{figure}
    \centering
    \includegraphics[width=0.95\linewidth]{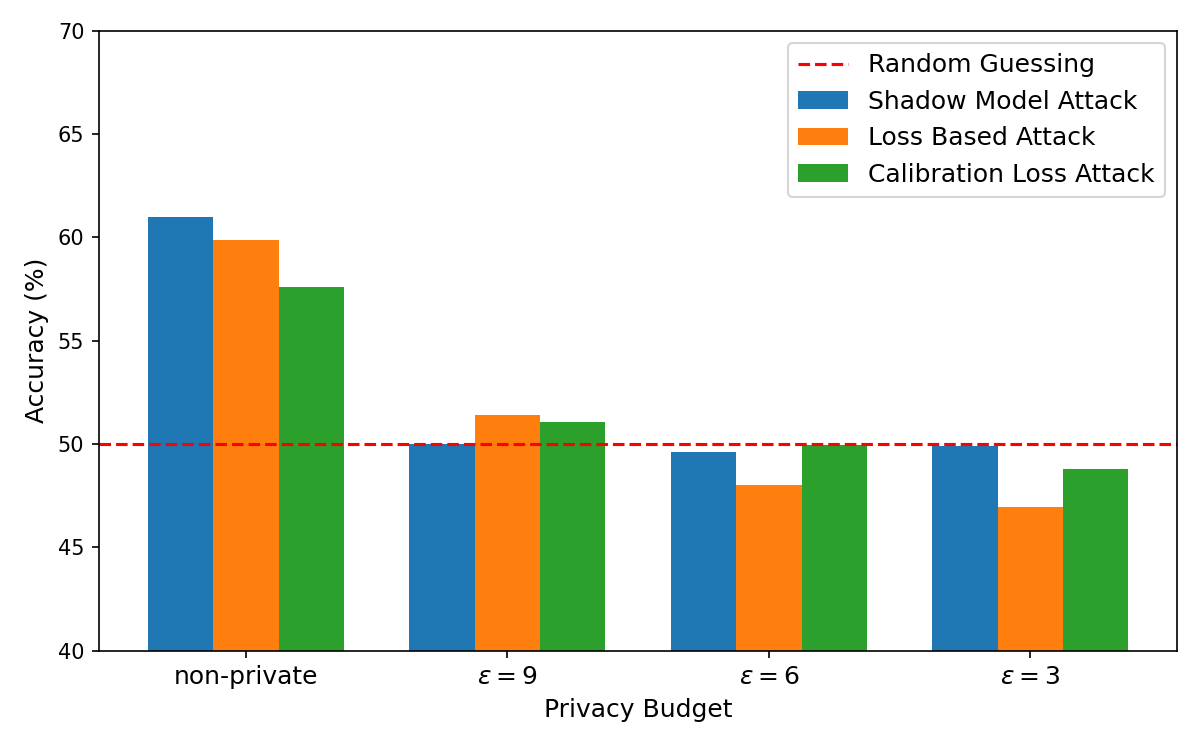}
    \caption{Success rates of three different membership inference attacks in non-private and private settings. We include the random guessing $(0.5)$ as a baseline. Accuracies that are closed to or under the baseline indicates successful protection against the attacks.}
    \Description{Success rates of three different membership inference attacks in non-private and private settings. We include the random guessing $(0.5)$ as a baseline. Accuracies that are closed to or under the baseline indicates successful protection against the attacks.}
    \label{fig:mia}
\end{figure}

Figure~\ref{fig:mia} shows the success rate of each attack, averaged over $5$ random seeds. In the non-private setting, the model exhibits a clear vulnerability to membership leakage, with attack accuracies reaching up to $0.61$, significantly outperforming the $0.5$ random guessing threshold. However, in private settings, the attack accuracy for all three methods collapses to approximately $0.5$, even at a less private setting with $\epsilon=9$. As the privacy budget is further tightened to $\epsilon=6$ and $\epsilon=3$, the attacker's advantage remains random. This convergence toward the random guessing baseline demonstrates that FedPower effectively hides the influence of individual training samples. The FedPower algorithm renders the model's output distribution indistinguishable between members and non-members and provides robust empirical validation of its privacy guarantees.

We report the true positive rate and false positive rate for each membership inference attack in Figure~\ref{tpr}. We observe a positive trend in non-private settings, where the curve is the highest across all three attacks. These results demonstrate the greatest success for the attacker in non-private settings, where the final model memorizes some aspects of the training set and makes it easy for attackers to distinguish members from non-members.

However, the curves begin to drop as DP noise is introduced. Regardless of any privacy budget value, both shadow model attack (Figure~\ref{tpr:shadow}) and calibration loss attack (Figure~\ref{tpr:calibration}) show immediate collapses to the random guessing baseline. The loss based attack shows a clearer separation between $\epsilon = 9$, $\epsilon = 6$, and $\epsilon = 3$ (Figure~\ref{tpr:loss}).
For the configurations with stronger privacy guarantees ($\epsilon = 6$ and $\epsilon = 3$), the loss based attack is successfully neutralized, with the curves converging closely toward the diagonal random guessing line. Even at a more relaxing budget of $\epsilon = 9$, the defense provides a substantial reduction in the attacker’s true positive rate compared to the non-private model. The evaluation of the membership inference attack demonstrates FedPower's robust protection as the privacy budget $\epsilon$ is tightened.

\begin{figure*}[h]
\centering
\begin{subfigure}{0.32\textwidth}
\centering
  \includegraphics[width=\textwidth]{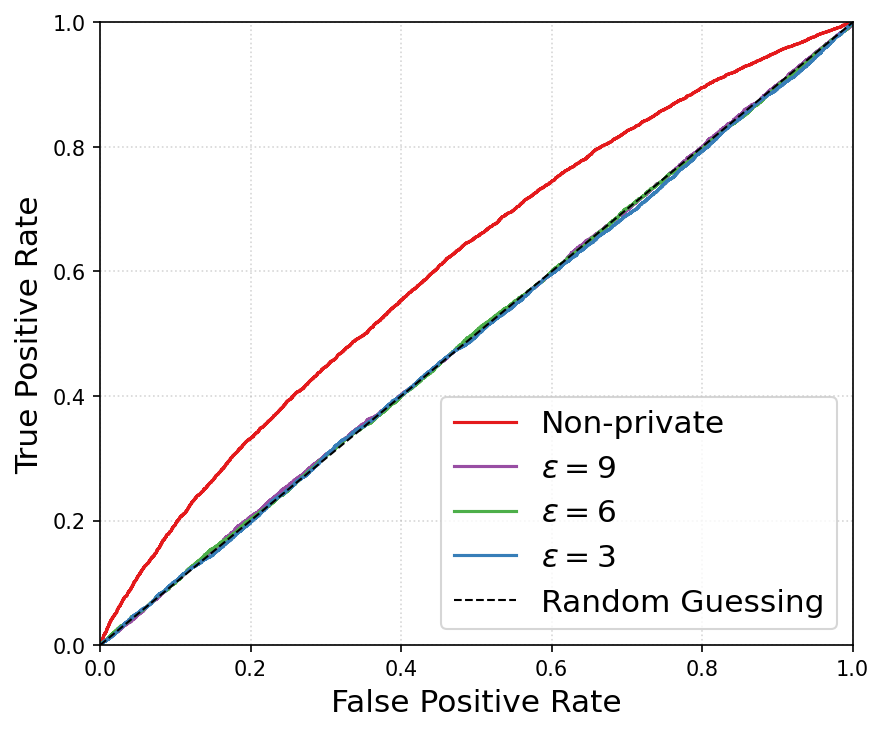}
  \caption{Shadow Model Attack.}
  \label{tpr:shadow}
\end{subfigure}
\begin{subfigure}{0.32\textwidth}
\centering
  \includegraphics[width=\textwidth]{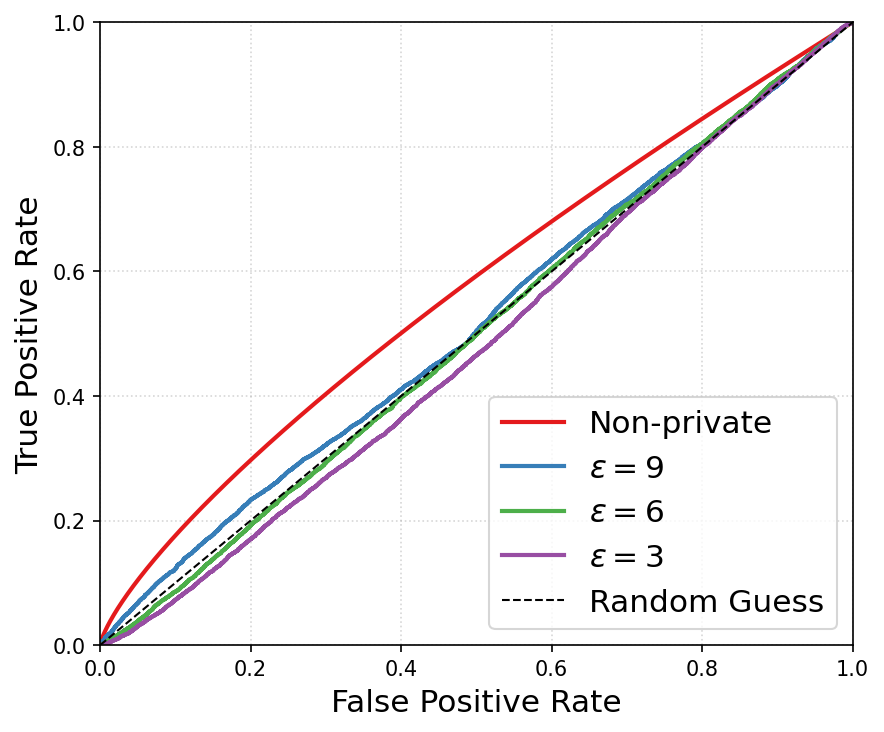}
  \caption{Loss Based Attack.}
  \label{tpr:loss}
\end{subfigure}
\begin{subfigure}{0.32\textwidth}
\centering
  \includegraphics[width=\textwidth]{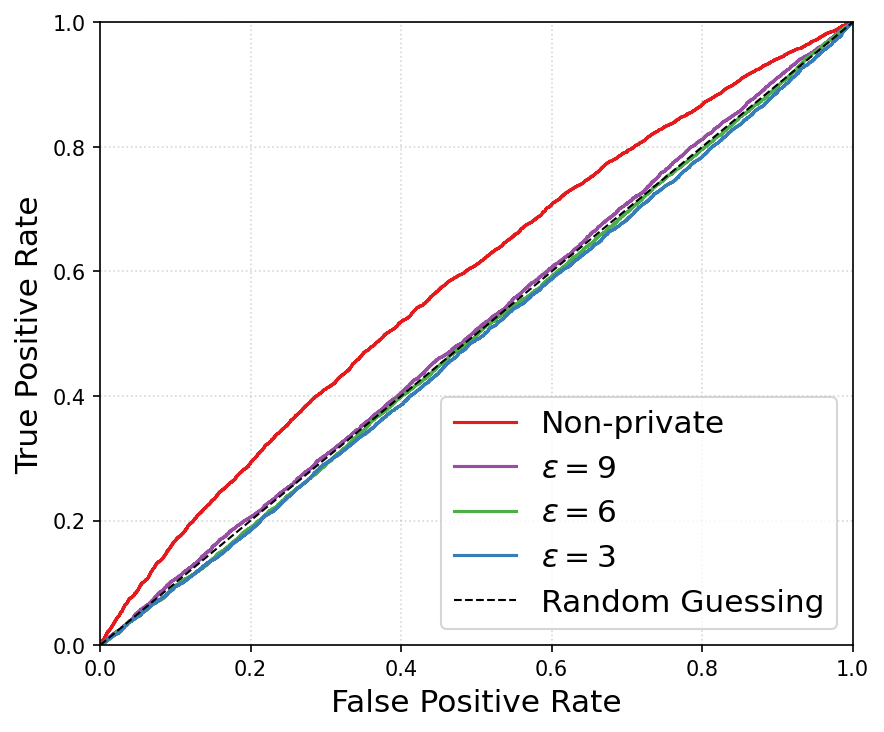}
  \caption{Calibration Loss Attack.}
  \label{tpr:calibration}
\end{subfigure}
\caption{True positive rate vs. False positive rate of each membership inference attack across all privacy budgets.}
\Description{True positive rate vs. False positive rate of each membership inference attack across all privacy budgets.}
\label{tpr}
\end{figure*}

\section{Conclusion and Future Works}
\label{sec:conclude}

In this work, we introduced FedPower, a differentially private federated learning framework that addresses the mathematical mismatch inherent in standard DP-FedLoRA. FedPower entirely eliminates the cross-term aggregation errors that typically compound with DP noise by shifting the server-side aggregation to the merged, full-rank equivalent space. To safely project these aggregated updates back into a communication-efficient low-rank space, we proposed PowerDP, a novel lightweight differentially private refactorization mechanism. By embedding calibrated Gaussian noise directly into the simultaneous subspace iteration sequence, PowerDP successfully mitigates the negative effect of DP noise seen in standard output perturbation methods. Our rigorous theoretical analyses established tight sensitivity bounds, proving that FedPower satisfies both sample-level and client-level differential privacy. Extensive empirical evaluations on the GLUE benchmark demonstrated that FedPower significantly outperforms existing DP-FedLoRA baselines in model utility while maintaining negligible server overhead and providing robust defense against membership inference attacks.

We primarily focused on establishing a robust and highly efficient framework for cross-silo federated learning in this work. Building upon the rigorous client-level DP guarantees already established in our theoretical analysis, a natural direction for future research is exploring FedPower in cross-device federated learning settings. Future studies could empirically investigate the framework's scalability in these massive edge environments and explore what, if any, further adaptations might be beneficial. Additionally, while FedPower demonstrates strong utility under homogeneous (IID) federated data distributions, future work can explore its optimization dynamics under heterogeneous (non-IID) data partitions, further broadening the applicability of private LLM fine-tuning across diverse, real-world decentralized systems.

\bibliographystyle{ACM-Reference-Format}
\bibliography{references}

\newpage

\appendix

\section{Additional Details on Experimental Setup}
\label{app:exp}

Table~\ref{tab:dataset} details the information of the datasets used in the experiments. Each training dataset is split evenly (IID) across $6$ clients for local training. At the end of each global round, the server evaluate the current global model on the test dataset, and we report the final test accuracy at the end of the FL training framework.

\begin{table}[H]
    \caption{Overview of datasets used in the experiments.}
    \label{tab:dataset}
    \centering
    \small
    \begin{tabular}{| c | c c c |}
        \hline
        Dataset &  Number of & Training & Testing \\
        & classes & set size & set size \\
        \hline\hline
        {MNLI (matched)} & \multirow{2}{*}{3} & \multirow{2}{*}{392{,}702} & 9,815 \\
        {MNLI (mismatched)} & & & 9,832 \\
        SST-2 & 2 & 67,349 & 872 \\
        QQP & 2 & 363,846 & 40,430 \\
        QNLI & 2 & 104,743 & 5,463 \\
        \hline
    \end{tabular}
\end{table}

Algorithm~\ref{app:power} describes the standard non-private power iteration that is used in FedPower with input and output perturbation in Section~\ref{sec:exp:subsec:powerdp}. The algorithm starts with a random Gaussian matrix $Q$, and iteratively refines the basis by alternating between projecting the weight matrix onto the current subspace and orthonormalizing the result to prevent numerical instability. This process forces the subspace to align with the directions of greatest variance in the data. By the end of $k$ iterations, the algorithm produces an orthonormal basis $A$ and a corresponding projection $B$, which together capture the most significant structural information of the original matrix $W$ at the specified rank $r$.

\begin{algorithm}[H]
    \caption{Power Iteration (non-private) \cite{stewart1981simultaneous}}
    \label{app:power}
    \begin{algorithmic}[1]
    \STATE {\textbf{Input:}} $W \in \mathbb{R}^{m \times n}$, number of iterations $k$, rank $r$.
    \STATE Initialize $Q \in \mathbb{R}^{r \times n}$ from Gaussian distribution.
    \FOR {$k$ iterations}
        \STATE $P \gets W Q^T$.
        \STATE $P \gets$ Orthonormalize columns of $P$.
        \STATE $Q \gets P^T W$.
    \ENDFOR
    \STATE $A \gets$ Orthonormalize rows of $Q$.
    \STATE $B \gets W A^T$.
    \RETURN $A, B$.
\end{algorithmic}
\end{algorithm}

We provide the detailed algorithms for the two baselines used in the experiments: FedLoRA in Algorithm~\ref{alg:fedlora:dp} and FFA-LoRA in Algorithm~\ref{alg:ffa-lora:dp}.

\begin{algorithm}[H]
\caption{FedLoRA \cite{zhang2024towards}}
\label{alg:fedlora:dp}
\begin{algorithmic}[1]
    \STATE {\textbf{Input:}} number of global rounds $T$, number of local rounds $L$, global weight $W^0 \in \mathbf{R}^{m \times n}$, rank $r$, learning rate $\eta$, noise scale $\sigma$, clipping threshold $C$.
    \STATE \textcolor{blue}{// Server initializes global LoRA modules.}
    \STATE $A^0 \sim \mathcal{N}(0, \sigma^2)^{r \times n}$.
    \STATE $B^0 \gets \mathbf{0}_{m \times r}$.
    \FOR {global round $t \gets 1 \ldots T$}
        \STATE Server subsamples client set $\mathcal{C}^t$ with sampling rate $q_c$.
        \FOR {client $i \in \mathcal{C}$ in parallel}
            \STATE \textcolor{blue}{// Client initialize local LoRA modules.}
            \STATE $A^{t,0}_i \gets A^{t-1}$
            \STATE $B^{t,0}_i \gets B^{t-1}$.
            \FOR {local round $l \gets 1 \ldots L$}
                \STATE Client subsamples batch $\mathcal{B}^{t,l}$ with sampling rate $q_s$.
                \STATE \textcolor{blue}{// Client computes loss.}
                \STATE $\mathcal{L}_i^{t,l} \gets \frac{1}{| \mathcal{B}^{t,l} |} \sum_{(x,y)\in \mathcal{B}^{t,l}} \ell(A^{t,l-1}_i, B^{t,l-1}_i, x, y)$.
                \STATE \textcolor{blue}{// Client updates local LoRA modules.}
                \STATE $A^{t,l}_i \gets A^{t,l-1}_i - \eta \nabla_a \mathcal{L}_i^{t,l}$
                \STATE $B^{t,l}_i \gets B_i^{t,l-1} - \eta \nabla_b \mathcal{L}^{t,l}_i$.
            \ENDFOR
            \STATE Clients send $A^{t,L}_i, B^{t,L}_i$ to server.
        \ENDFOR
        \STATE Clip $A^{t,L}_i$ and $B^{t,L}_i$ by $C$.
        \STATE \textcolor{blue}{// Server aggregates and updates global LoRA modules.}
        \STATE $\displaystyle A^t \gets  \frac{1}{|\mathcal{C}|} \sum_{i\in \mathcal{C}} A_i^{t,L}$.
        \STATE $\displaystyle B^t \gets \frac{1}{|\mathcal{C}|} \sum_{i\in \mathcal{C}} B_i^{t,L}$.
        \STATE \textcolor{blue}{// Server adds DP noises.}
        \STATE $\displaystyle A^t \gets A^t + \mathcal{N}(0, \sigma^2 C^2)$.
        \STATE $\displaystyle B^t \gets B^t + \mathcal{N}(0, \sigma^2 C^2)$.
    \ENDFOR
    \STATE \textcolor{blue}{// Server merges LoRA for final release.}
    \STATE $W^T \gets W^0 + B^T A^T$.
    \STATE {\textbf{Output:}} $W^T$. 
\end{algorithmic}
\end{algorithm}

As described in Algorithm~\ref{alg:fedlora:dp}, the server initializes two small low-rank matrices, $A$ and $B$. During each global round, selected clients download these global LoRA modules, perform local training over $L$ rounds to minimize their specific loss, and then return the updated modules to the server. To ensure privacy, the server clips the client updates to a threshold $C$ to bound their influence and injects Gaussian noise to the clipped modules. Most importantly, the server then aggregates each clipped LoRA modules independently. Once the global training rounds are complete, the final low-rank product $B^T A^T$ is merged back into the original weights, resulting in a fine-tuned model that benefits from decentralized data without exposing individual client updates.

\begin{algorithm}[t]
\caption{FFA-LoRA \cite{sun2024improving}}
\label{alg:ffa-lora:dp}
\begin{algorithmic}[1]
    \STATE {\textbf{Input:}} number of global rounds $T$, number of local rounds $L$, global weight $W^0 \in \mathbf{R}^{m \times n}$, rank $r$, learning rate $\eta$, noise scale $\sigma$, clipping threshold $C$.
    \STATE \textcolor{blue}{// Server initializes global LoRA modules.}
    \STATE $A \sim \mathcal{N}(0, \sigma^2)^{r \times n}$.
    \STATE $B^0 \gets \mathbf{0}_{m \times r}$.
    \FOR {global round $t \gets 1 \ldots T$}
        \STATE Server subsamples client set $\mathcal{C}^t$ with sampling rate $q_c$.
        \FOR {client $i \in \mathcal{C}$ in parallel}
            \STATE \textcolor{blue}{// Client initialize local LoRA module.}
            \STATE $B^{t,0}_i \gets B^{t-1}$.
            \FOR {local round $l \gets 1 \ldots L$}
                \STATE Client subsamples batch $\mathcal{B}^{t,l}$ with sampling rate $q_s$.
                \STATE \textcolor{blue}{// Client computes loss.}
                \STATE $\mathcal{L}_i^{t,l} \gets \frac{1}{| \mathcal{B}^{t,l} |} \sum_{(x,y)\in \mathcal{B}^{t,l}} \ell(A, B^{t,l-1}_i, x, y)$.
                \STATE \textcolor{blue}{// Client updates local LoRA modules.}
                \STATE $B^{t,l}_i \gets B_i^{t,l-1} - \eta \nabla_b \mathcal{L}^{t,l}_i$.
            \ENDFOR
            \STATE Clients send $B^{t,L}_i$ to server.
        \ENDFOR
        \STATE Clip $B^{t,L}_i$ by $C$.
        \STATE \textcolor{blue}{// Server aggregates and updates global LoRA modules.}
        \STATE $\displaystyle B^t \gets \frac{1}{|\mathcal{C}|} \sum_{i\in \mathcal{C}} B_i^{t,L}$.
        \STATE \textcolor{blue}{// Server adds DP noises.}
        \STATE $\displaystyle B^t \gets B^t + \mathcal{N}(0, \sigma^2 C^2)$.
    \ENDFOR
    \STATE \textcolor{blue}{// Server merges LoRA for final release.}
    \STATE $W^T \gets W^0 + B^T A$.
    \STATE {\textbf{Output:}} $W^T$. 
\end{algorithmic}
\end{algorithm}

FFA-LoRA is structurally identical to FedLoRA in the use of federated averaging, client-side local training, and server-side DP (clipping and noise injection). However, it introduces a main difference by freezing the $A$ module throughout the entire training process. While FedLoRA iteratively updates both $A$ and $B$, FFA-LoRA treats $A$ as a static Gaussian projection, requiring clients to only compute gradients for and transmit the $B$ module. Upon receiving the LoRA $B$ modules from clients, the server clips, adds noise and averages the $B$ modules accordingly.

\end{document}